\documentclass[12pt,onecolumn,twoside]{IEEEtran}
\usepackage{mathpazo}
\usepackage{times}

\usepackage{color}
\usepackage{amsthm}
\usepackage{amssymb}
\usepackage{amsmath}
\usepackage{bm}

\usepackage{yfonts}
\renewcommand{\mathbf}[1]{{\bm{#1}}}     

\title{On Nearly Perfect Covering Codes}

\date{\today}
\author{\textbf{Avital Boruchovsky}, \textbf{Tuvi Etzion}, and \textbf{Ron M. Roth}\\
{\small Computer Science Department, Technion, Israel Institute of Technology, Haifa 3200003, Israel}
}

\theoremstyle{definition} 
\newtheorem{theorem}{\indent Theorem}

\newtheorem{lemma}[theorem]{\indent Lemma}
\newtheorem{corollary}[theorem]{\indent Corollary}

\theoremstyle{remark}
\newtheorem{remark}{\indent Remark}
\theoremstyle{definition}
\newtheorem{example}{\indent Example}

\newcommand{\cC}{{\mathcal{C}}}

\newcommand{\cE}{{\mathcal{E}}}

\newcommand{\cP}{{\mathcal{P}}}
\newcommand{\cQ}{{\mathcal{Q}}}

\newcommand{\cX}{{\mathcal{X}}}
\newcommand{\cY}{{\mathcal{Y}}}

\newcommand{\abs}[1]{\left|#1\right|}

\renewcommand{\le}{\leqslant}
\renewcommand{\leq}{\leqslant}
\renewcommand{\ge}{\geqslant}

\newcommand{\F}{{\mathbb{F}}}

\newcommand{\dS}{{\mathbb{S}}}

\newcommand{\bldzero}{{\mathbf{0}}}

\newcommand{\bldA}{{\mathbf{A}}}
\newcommand{\bldB}{{\mathbf{B}}}
\newcommand{\bldc}{{\mathbf{c}}}
\newcommand{\blde}{{\mathbf{e}}}
\newcommand{\bldu}{{\mathbf{u}}}
\newcommand{\bldv}{{\mathbf{v}}}
\newcommand{\bldx}{{\mathbf{x}}}
\newcommand{\bldy}{{\mathbf{y}}}
\newcommand{\bldz}{{\mathbf{z}}}
\newcommand{\Code}{{\mathsf{C}}}
\renewcommand{\AA}{{\mathsf{A}}}
\newcommand{\BB}{{\mathsf{B}}}

\newcommand{\GG}{{\mathsf{G}}}
\newcommand{\YY}{{\mathsf{Y}}}
\newcommand{\TA}{{\texttt{\large A}}}
\newcommand{\TB}{{\texttt{\large B}}}
\newcommand{\TC}{{\texttt{\large C}}}
\newcommand{\Tone}{{\mathrm{I}}}
\newcommand{\Ttwo}{{\mathrm{II}}}
\newcommand{\Type}[1]{{Type~${#1}$}}
\newcommand{\Ball}{{\mathfrak{B}}}
\newcommand{\Sphere}{{\mathcal{S}}}
\newcommand{\Support}{{\mathsf{Supp}}}

\newcommand{\Int}[1]{{\left[{#1}\right]}}
\newcommand{\Realfield}{{\mathbb{R}}}
\newcommand{\Ring}{{\mathfrak{R}}}
\newcommand{\distance}{{\mathsf{d}}}
\newcommand{\weight}{{\mathsf{w}}}
\newcommand{\Emph}[1]{{\textbf{\emph{#1}}}}

\makeatletter
\renewcommand{\@endtheorem}{\endtrivlist}
\makeatother

\begin{document}

\maketitle

\begin{abstract}
Nearly perfect packing codes are those codes that meet
the Johnson upper bound on the size of error-correcting codes.
This bound is an improvement to the sphere-packing bound.
A related bound for covering codes is known as the van Wee bound.
Codes that meet this bound will be called nearly perfect covering codes.
In this paper, such codes with covering radius one will be considered.
It will be proved that these codes can be partitioned into
three families depending on the smallest distance between neighboring
codewords. Some of the codes contained in these families
will be completely characterized.
Other properties of these codes will be considered too.
Construction for codes for each such family will be presented,
the weight distribution and the distance distribution of codes from
these families are characterized. Finally, extended nearly perfect
covering code will be considered and unexpected equivalence classes
of codes of the three types will be defined based on the extended codes.
\end{abstract}

\section{Introduction}
\label{sec:intro}

Perfect codes are among the most fascinating structures
in coding theory. They meet the well-known sphere-packing bound,
yet they are very rare. Therefore, there have been many attempts
to find either packing or covering codes that are ``almost perfect.''
One class of such covering codes is the topic of this paper.

All the codes in this work are over the binary field $\F_2$,
and by an \Emph{$(n,M)$ code} (of length~$n$ and size~$M$)
we mean a subset $\cC \subseteq \F_2^n$ of size $\abs{\cC} = M$.
For integers $\ell \le m$,
we use the notation $\Int{\ell:m}$
for the integer interval $\{ \ell, \ell{+}1, \ldots, m \}$,
with $\Int{m}$ standing for $\Int{1:m}$.

A \Emph{translate} of an $(n,M)$ code~$\cC$ is the set
\[
\blde + \cC \triangleq \{ \blde + \bldc ~:~ \bldc \in \cC \}~,
\]
where $\blde \in \F_2^n$ (and addition is over $\F_2$).
When the all-zero word, $\bldzero$, is a codeword in~$\cC$
we say that the code is \Emph{zeroed}.
A translate $\blde + \cC$ with $\blde \in \cC$ is a zeroed code.
A translate $\blde + \cC$ with $\blde \not\in \cC$ is
a non-zeroed translate.

The \Emph{(Hamming) distance} between two words
$\bldx, \bldy \in \F_2^n$ will be denoted by $\distance(\bldx,\bldy)$,
and $\weight(\bldx)$ will denote the \Emph{weight} of~$\bldx$,
i.e., the size of the support, $\Support(\bldx)$, of~$\bldx$
(the notation extends to integer vectors as well).
The \Emph{radius-$t$ ball} centered at a word $\bldx \in \F_2^n$
is denoted by
\[
\Ball_t(\bldx)
\triangleq \{ \bldy \in \F_2^n ~:~ \distance(\bldx,\bldy) \le t \}
\]
(where, for simplicity of notation, we make the dependence on~$n$
implicit).
We also define the boundary (sphere)
\[
\partial \Ball_t(\bldx)
\triangleq \{ \bldy \in \F_2^n ~:~ \distance(\bldx,\bldy) = t \}
\]
and
\[
\Sphere_t \triangleq
\partial \Ball_t(\bldzero)
= \{ \bldy \in \F_2^n ~:~ \weight(\bldy) = t \}~.
\]
The words in $\partial \Ball_t(\bldx)$ will be called
the \Emph{$t$-neighbors} of~$\bldx$.

The \Emph{minimum distance} of an $(n,M)$ code~$\cC$
is the smallest distance between any two distinct codewords
in~$\cC$, and the distance of a word $\bldx \in \F_2^n$ from~$\cC$
is defined by
$\distance(\bldx,\cC) = \min_{\bldc \in \cC} \distance(\bldx,\bldc)$.
The \Emph{covering radius} of~$\cC$
is defined by
\[
R = \max_{\bldx \in \F_2^n} \distance(\bldx,\cC)~,
\]
and we say that a code~$\cC$ is \Emph{$R$-covering}
if its covering radius is at most~$R$.
When $\cC$~is a linear code over $\F_2$ (of dimension $\log_2 M$)
and $H$ is any full-rank $r \times n$ parity-check matrix of~$\cC$
over $\F_2$ (where $r = n - \log_2 M$),
then the covering radius of~$\cC$ equals the smallest~$R$ such that
every vector in $\F_2^r$ can be expressed as
a linear combination (over $\F_2$) of~$R$ columns of~$H$.

For an $(n,M)$~code~$\cC$ with minimum distance $2R+1$
we have the \Emph{sphere-packing bound}
\begin{equation}
\label{eq:sphere-packing}
M \cdot \sum_{i=0}^R \binom{n}{i} \le 2^n ~,
\end{equation}
and if~$\cC$ has covering radius~$R$
we have the \Emph{sphere-covering bound}
\begin{equation}
\label{eq:sphere-covering}
M \cdot \sum_{i=0}^R \binom{n}{i} \ge 2^n ~.
\end{equation}
Perfect codes meet both bounds.

The sphere-packing bound for an $(n,M)$~code with
minimum distance $2R+1$ was improved by Johnson~\cite{Joh62} to
\begin{equation}
\label{eq:johnson}
M \cdot \left(
\sum_{i=0}^R \binom{n}{i} + \frac{\binom{n}{R}}{\left\lfloor
\frac{n}{R+1} \right\rfloor}
\left(\frac{n-R}{R+1} - \left\lfloor \frac{n-R}{R+1} \right\rfloor
\right) \right) \le 2^n ~,
\end{equation}
and a code that meets this bound is called
a \Emph{nearly perfect (packing) code}.
When $R+1$ divides $n-R$, this bound coincides
with the sphere-packing bound.
Codes that meet the bound~(\ref{eq:johnson}) were considered
in~\cite{GoSn72}, \cite{Lind77}.
There are two families of nontrivial codes that
are nearly perfect yet not perfect.
One family is the set of \Emph{shortened Hamming codes}.
A second family consists of the \Emph{punctured Preparata codes}.
These codes were first found by Preparata~\cite{Pre68} and later others
found many inequivalent codes with
the same parameters~\cite{BLW83}, \cite{Kan83}.
Moreover, these codes are very important in constructing other codes,
e.g., see~\cite{EtGr93}. A comprehensive work on perfect
codes and related codes can be found in~\cite{Etz22}.
For $R$-covering codes, an improvement on the sphere-covering bound,
akin to the Johnson bound, was
presented by van Wee~\cite{vWee88}.
A simplified version of his bound
was presented by Struik~\cite{Str94} and takes the form
\begin{equation}
\label{eq:struik}
M \cdot \left( \sum_{i=0}^R \binom{n}{i}
- \frac{\binom{n}{R}}{\left\lceil \frac{n-R}{R+1} \right\rceil}
\left( \left\lceil \frac{n+1}{R+1}
\right\rceil - \frac{n+1}{R+1} \right) \right) \ge 2^n ~.
\end{equation}
When $R+1$ divides $n+1$, the bound~(\ref{eq:struik}) coincides with
the sphere-covering bound.
A~code that meets this bound will be called
a \Emph{nearly perfect covering code}.
One can easily see the similarity and the difference between
the bounds~(\ref{eq:johnson}) and~(\ref{eq:struik}).
For even~$n$ and $R=1$, the bound~(\ref{eq:struik}) becomes
\begin{equation}
\label{eq:vanWee}
M \ge \frac{2^n}{n}~.
\end{equation}

Except for perfect codes and some trivial codes,
$R = 1$ is the only radius for which we currently know
of codes that meet the bound~(\ref{eq:struik}); from~(\ref{eq:vanWee}),
these codes have length $n = 2^r$ and size $M = 2^{2^r-r}$,
for some positive integer~$r$. A code with these parameters will be
called a \Emph{a nearly perfect $1$-covering code}
(in short, \Emph{NP1CC}). In the case of linear codes,
there is a very simple characterization
of NP1CCs, as we show in the next example.

\begin{example}
\label{ex:linearNP1CC}
Let $\cC$ be an $(n{=}2^r,M{=}2^{n-r})$ linear code over~$\F_2$
and let~$H$ be any full-rank $r \times n$
parity-check matrix of~$\cC$ over~$\F_2$. Then~$\cC$ is $1$-covering
(and, hence, is an NP1CC), if and only if
each nonzero vector in $\F_2^r$ appears as a column in~$H$.
Thus, there are two possible cases. The first is when
the columns of~$H$ range over all the vectors of~$\F_2^r$
(including the all-zero vector); the second case is similar,
except that the all-zero column is replaced by
some nonzero vector of~$\F_2^r$.\qed
\end{example}

In this work, we consider the structure of
general (not necessarily linear) NP1CCs.
In Section~\ref{sec:struct}, we prove that in any NP1CC~$\cC$,
each codeword $\bldc \in \cC$ has a unique other codeword $\bldc'$
in $\Ball_2(\bldc)$; this, in turn, induces a partition
of the code~$\cC$ into pairs $\{ \bldc, \bldc' \}$.
Based on this property, in Section~\ref{sec:construct}
we classify NP1CCs into three types:
\begin{itemize}
\item
\Type{\TA} codes, in which the codewords in every pair
$\{ \bldc, \bldc' \}$ are at distance~$1$ apart
(the first case in Example~\ref{ex:linearNP1CC} belongs to this type),
\item
\Type{\TB} codes, in which the codewords in every such pair
are at distance~$2$ apart
the second case in the example is of this type),
and
\item
\Type{\TC} codes (which are all the remaining NP1CCs).
\end{itemize}
We study the properties of these types (especially of \Type{\TA} codes)
and present constructions of codes for each type.
In Section~\ref{sec:weightDist}, we consider the weight and distance
distributions of NP1CCs and, in particular, we prove that there are
exactly two weight distributions for all the codes and two other
weight distributions for all their translates.
Moreover, we show that \Type{\TA} and \Type{\TB} codes
are distance invariant.
In Section~\ref{sec:balanced},
we concentrate on a class of \Type{\TA} codes
in which the number of codeword pairs $\{ \bldc, \bldc' \}$ that differ
only on any given coordinate is the same for all coordinates.
Extended NP1CCs are discussed in Section~\ref{sec:equiv},
where we prove that we can define equivalence classes for NP1CCs
of the three types via the punctured codes of the extended code.
A conclusion and a few problems for future research are presented
in Section~\ref{sec:conclude}.

\section{Structure of NP1CCs}
\label{sec:struct}

In this section, we examine the structure of NP1CCs.

Let $\cC$ be an $(n,M)$~code.
Given a word $\bldx \in \F_2^n$, we say that a codeword
$\bldc \in \cC$ \Emph{covers}~$\bldx$ if $\bldc \in \Ball_1(\bldx)$.
Clearly, if $\cC$ is $1$-covering then
every word $\bldx \in \F_2^n$ is covered by
at least one codeword of~$\cC$.
The \Emph{over-covering} of
a subset $\cY \subseteq \F_2^n$
(with respect to a $1$-covering code~$\cC$) is defined by
\[
\sum_{\bldy \in \cY} (\abs{\Ball_1(\bldy) \cap \cC} - 1)
= \Bigl( \sum_{\bldy \in \cY} \abs{\Ball_1(\bldy) \cap \cC} \Bigr)
- \abs{\cY}~.
\]
Thus, while each word in~$\cY$ is covered by at least
one codeword of~$\cC$, the over-covering of~$\cY$ measures
how many \Emph{additional} codewords cover each one of
the words in~$\cY$.

The following lemma follows from the analysis
of Struik in~\cite{Str94} (although, as stated, it does not appear
explicitly there).

\begin{lemma}[\cite{Str94}]
\label{lem:cover_ball_x}
Let $\cC$ be an $(n,M)$~NP1CC
and let $\bldx \in \F_2^n \setminus \cC$ be a non-codeword.
Then $\Ball_1(\bldx)$ contains exactly
one word that is covered by two codewords of~$\cC$
and no word that is covered by more than two codewords of~$\cC$.
\end{lemma}

\begin{proof}
We provide the steps of the proof through pointers to~\cite{Str94}.
It follows from Eq.~(6) therein that for each non-codeword $\bldx \in \F_2^n \setminus \cC$,
the over-covering of the ball $\Ball_1(\bldx)$ is at least~$1$.
Denoting by~$\epsilon$ the average of these over-coverings,
we then get that $\epsilon \ge 1$ (see~Eq.~(7) in~\cite{Str94}).
Then, equality in the Van Wee bound (Eq.~(9) in~\cite{Str94})
forces the equality $\epsilon = 1$, which means that
the over-covering of each ball $\Ball_1(\bldx)$ must be exactly~$1$.
\end{proof}

\begin{corollary}
\label{cor:cover_nonC}
Let $\cC$ be an $(n,M)$~NP1CC.
For every non-codeword $\bldx \in \F_2^n \setminus \cC$,
\[
\abs{\Ball_1(\bldx) \cap \cC} \le 2~.
\]
\end{corollary}

A non-codeword $\bldx \in \F_2^n \setminus \cC$
for which $\abs{\Ball_1(\bldx) \cap \cC} = 2$
will be called a \Emph{midword}.

While midwords differ from the remaining non-codewords
in the size of the intersection $\Ball_1(\bldx) \cap \cC$,
those sizes become the same if we look at balls of radius~$2$.
This property, which we prove in the next theorem, will be
instrumental in Section~\ref{sec:weightDist}
for deriving the weight and distance distributions of NP1CCs.

\begin{theorem}
\label{thm:radius2}
Let~$\cC$ be an $(n,M)$~NP1CC.
For every non-codeword $\bldx \in \F_2^n \setminus \cC$,
\[
\abs{\Ball_2(\bldx) \cap \cC} = \frac{n}{2} + 1~.
\]
\end{theorem}

\begin{proof}
Consider first the case where $\bldx$ is a midword.
By Lemma~\ref{lem:cover_ball_x}, no other word in $\Ball_1(\bldx)$
is covered by two codewords;
namely, the set of $1$-neighbors of~$\bldx$
consists of two codewords $\bldc_1$ and $\bldc_2$
(none of which is a $1$-neighbor of a codeword)
and $n-2$ non-codewords
$\bldy_1, \bldy_2, \ldots, \bldy_{n-2}$ (none of which is a midword).
Each $\bldy_i$, in turn, is covered by a unique codeword
(which belongs to $\partial \Ball_2(\bldx)$).
Conversely, each codeword in $\partial \Ball_2(\bldx)$ covers exactly
two words among the $\bldy_i$'s (and none of the codewords $\bldc_1$
and $\bldc_2$).
We conclude that $\Ball_2(\bldx)$ contains exactly $n/2 + 1$ codewords:
the codewords $\bldc_1$ and $\bldc_2$, and $n/2 - 1$ codewords
that cover the $\bldy_i$'s.

We next turn to the case where $\bldx$ is not a midword.
One (and only one) of the $1$-neighbors of $\bldx$ is
a codeword, $\bldc$, and, by Lemma~\ref{lem:cover_ball_x},
there is a unique $1$-neighbor~$\bldy_0$ of~$\bldx$
that is covered by two codewords. We distinguish between two cases.

\Emph{Case 1:} $\bldy_0 = \bldc$.
The $n-1$ remaining $1$-neighbors of~$\bldx$ are non-codewords
$\bldy_1, \bldy_2, \ldots, \bldy_{n-1}$,
and each $\bldy_i$ (including $\bldy_0$) is covered by
a unique codeword in $\partial \Ball_2(\bldx)$.
Conversely, each codeword in $\partial \Ball_2(\bldx)$
covers exactly two words among the $\bldy_i$'s.
These $n/2$~codewords, along with $\bldc$, are (all)
the $n/2 + 1$~codewords in $\Ball_2(\bldx)$.

\Emph{Case 2:} $\bldy_0 \ne \bldc$,
namely, $\bldy_0$ is a midword,
which is covered by two codewords
$\bldc_1, \bldc_2 \in \partial \Ball_2(\bldx)$.
The $1$-neighbors of~$\bldx$ other than~$\bldc$ and~$\bldy_0$
are non-codewords
$\bldy_1, \bldy_2, \ldots, \bldy_{n-2}$, each covered by
a unique codeword (in $\partial \Ball_2(\bldx)$).
Conversely, each codeword in
$\partial \Ball_2(\bldx)$ covers exactly
two words among the $\bldy_i$'s
(where $\bldy_0$ is covered by two codewords).
We conclude that $\Ball_2(\bldx)$
contains exactly $n/2 + 1$~codewords:
(i)~the codeword~$\bldc$,
(ii)~the codewords $\bldc_1$ and~$\bldc_2$, which cover both $\bldy_0$
and two other $\bldy_i$'s, and
(iii)~$n/2-2$~codewords that cover the $n-4$ remaining $\bldy_i$'s.
\end{proof}

The next theorem is due to
Fort and Hedlund~\cite{FoHe58} and will be used in the proof
of our next lemma.

\begin{theorem}[\cite{FoHe58}]
\label{thm:cov_pairs_tri}
Let $\cX$ be an $(n,M)$ code whose codewords are all in $\Sphere_3$
and, in addition, every word in $\Sphere_2$ is covered by
at least one codeword in~$\cX$. Then
\[
\abs{\cX} \ge \left\lceil \frac{n}{3}
\left\lceil \frac{n-1}{2} \right\rceil \right\rceil ~.
\]
\end{theorem}

\begin{lemma}
\label{lem:noD3}
Let~$\cC$ be an $(n,M)$~NP1CC.
For every codeword $\bldc \in \cC$,
\[
\abs{\Ball_2(\bldc) \cap \cC} \ge 2~.
\]
\end{lemma}

\begin{proof}
The result is immediate when $n = 2$, so we assume hereafter
in the proof that $n = 2^r \ge 4$ and (by possibly translating
the code) that $\bldc = \bldzero$.
Suppose to the contrary that
$\Ball_2 (\bldzero) \cap \cC = \{ \bldzero \}$.
Then $\Sphere_1 \cap \cC = \Sphere_2 \cap \cC = \emptyset$
and, so, all the words in $\Sphere_2$
are covered (only) by codewords in $\Sphere_3 \cap \cC$.
By Theorem~\ref{thm:cov_pairs_tri} we then get that
$\abs{\Sphere_3 \cap \cC} \ge (n^2/2 + 1)/3$.
Now, each codeword in $\Sphere_3 \cap \cC$
covers three words in $\Sphere_2$
and, hence, the over-covering of $\Sphere_2$ (with respect to $\cC$)
satisfies
\[
\sum_{\bldy \in \Sphere_2} (\abs{\Ball_1(\bldy) \cap \cC} - 1)
= 3 \abs{\Sphere_3 \cap \cC} - \abs{\Sphere_2}
\ge \frac{n^2}{2} + 1 - \binom{n}{2} = \frac{n}{2} + 1~.
\]
On the other hand, by Corollary~\ref{cor:cover_nonC},
$\abs{\Ball_1(\bldy) \cap \cC} \in \{ 1, 2 \}$
for every $\bldy \in \Sphere_2$.
Hence, there are at least $n/2 + 1$ words $\bldy \in \Sphere_2$
for which $\abs{\Ball_1(\bldy) \cap \cC} = 2$, which means that
at least two of these words,
say $\bldy_1$ and $\bldy_2$, must have a `$1$' at the same position.
Let~$\bldx$ be the word in $\Sphere_1$ that has its (only)
`$1$' at that position.
Then $\Ball_1(\bldx)$ contains two words, $\bldy_1$ and $\bldy_2$,
each covered by two codewords,
thereby contradicting Lemma~\ref{lem:cover_ball_x}.
We thus conclude that $\abs{\Ball_2(\bldzero) \cap \cC} \ge 2$.
\end{proof}

The next theorem presents the counterpart of Theorem~\ref{thm:radius2}
for radius-$2$ balls that are centered at codewords of an NP1CC.

\begin{theorem}
\label{thm:adjacent_codewords}
Let~$\cC$ be an $(n,M)$~NP1CC.
For every codeword $\bldc \in \cC$,
\[
\abs{\Ball_2(\bldc) \cap \cC} = 2~.
\]
\end{theorem}

\begin{proof}
We consider the sum
\[
\rho = \sum_{\bldx \in \F_2^n} \abs{\Ball_2(\bldx) \cap \cC}~.
\]
Every codeword $\bldc \in \cC$ is counted in this sum
exactly $\abs{\Ball_2(\bldc)} = \abs{\Ball_2(\bldzero)}$ times; so,
\[
\rho = M \cdot \abs{\Ball_2(\bldzero)}
= M \cdot \left( \binom{n}{2} + n + 1 \right)
= M \cdot \left( \frac{n^2}{2} + \frac{n}{2} + 1 \right)~.
\]
Next, we write $\rho = \sigma + \tau$, where
\[
\sigma = \sum_{\bldx \in \F_2^n \setminus \cC}
\abs{\Ball_2(\bldx) \cap \cC}
\]
and
\begin{equation}
\label{eq:tau}
\tau = \sum_{\bldc \in \cC}
\abs{\Ball_2(\bldc) \cap \cC}~.
\end{equation}
By Theorem~\ref{thm:radius2} it follows that
\[
\sigma =
(2^n - M) \left( \frac{n}{2} + 1 \right)
= M \cdot (n-1) \left( \frac{n}{2} + 1 \right)
= M \cdot \left( \frac{n^2}{2} + \frac{n}{2} - 1 \right)
\]
and, so,
\[
\tau = \rho - \sigma = 2 M~.
\]
Now, by Lemma~\ref{lem:noD3},
each of the $M$~summands in~(\ref{eq:tau}) is at least~$2$;
hence, each of them must in fact be equal to~$2$.
\end{proof}

For any codeword $\bldc$ in an NP1CC~$\cC$, the unique other codeword
$\bldc'$ in $\Ball_2(\bldc)$ will be called
the \Emph{partner} of~$\bldc$.
A pair of partners $\{ \bldc, \bldc' \}$
in which~$\bldc$ and $\bldc'$ are at distance~$1$
(respectively, $2$) apart will be called
a \Emph{\Type{\Tone}} (respectively , \Emph{\Type{\Ttwo}}) \Emph{pair}.

\begin{corollary}
\label{cor:pair_parts}
Let~$\cC$ be an $(n{=}2^r,M{=}2^{n-r})$~NP1CC.
The codewords of~$\cC$ can be partitioned uniquely into
$M/2 = 2^{n-r-1}$ (unordered) pairs $\{ \bldc, \bldc' \}$,
where $\bldc$ and~$\bldc'$ are partners.
\end{corollary}

For a pair of partners $\{ \bldc, \bldc' \}$,
consider the ``capsule'' $\Ball_1(\bldc) \cup \Ball_2(\bldc')$.
We can distinguish between two types of capsules,
depending on whether the pair $\{ \bldc, \bldc' \}$
is of \Type{\Tone} or of \Type{\Ttwo}.
Interestingly, the two types of capsules have the same size, $2n$.
The midwords are precisely the words that belong to
the intersections $\Ball_1(\bldc) \cap \Ball_1(\bldc')$
when the pair is of \Type{\Ttwo}.

\begin{theorem}
\label{thm:covered_twice}
Let~$\cC$ be an $(n{=}2^r,M{=}2^{n-r})$~NP1CC.
There are exactly $M = 2^{n-r}$ words in $\F_2^n$ that are covered
by two codewords of~$\cC$ and no word is covered by more than
two codewords.
\end{theorem}

\begin{proof}
Each codeword of~$\cC$ covers $n+1$ words of $\F_2^n$
and, so,
\[
\sum_{\bldx \in \F_2^n} \abs{\Ball_1(\bldx) \cap \cC} = M(n+1)
= \abs{\F_2^n} + M~.
\]
The result follows from
Corollary~\ref{cor:cover_nonC}
and Theorem~\ref{thm:adjacent_codewords},
which imply that
$\abs{\Ball_1(\bldx) \cap \cC} \in \{ 1, 2 \}$
for every $\bldx \in \F_2^n$.
\end{proof}

The words in Theorem~\ref{thm:covered_twice}
that are covered by two codewords are
(i)~the midwords and
(ii)~the partners in \Type{\Tone} pairs.

We end this section by presenting sufficient conditions for
a code to be an NP1CC.

\begin{corollary}
\label{cor:nearly_part}
Let~$\cC$ an $(n,M)$ code where $M$~is even,
and suppose that $\cC$ can be partitioned into $M/2$
unordered pairs
$\{ \bldc, \bldc' \}$ where $\distance(\bldc,\bldc') \le 2$.
Suppose in addition that the respective $M/2$ capsules form
a partition of $\F_2^n$. Then $\cC$ is an NP1CC.
\end{corollary}

\begin{proof}
The code $\cC$ is $1$-covering since every word in $\F_2^n$
is contained in at least one capsule. And since the size
of each capsule is~$2n$ we get equality in~(\ref{eq:vanWee}).
\end{proof}

\begin{corollary}
\label{cor:sufficient}
Let~$\cC$ an $(n{=}2^r,M{=}2^{n-r})$ code
where $r$~is a positive integer.
Then~$\cC$ is an NP1CC,
if and only if $\abs{\Ball_2(\bldc) \cap \cC} = 2$
for every codeword $\bldc \in \cC$.
\end{corollary}

\begin{proof}
Theorem~\ref{thm:adjacent_codewords} establishes the ``only if'' part,
so we prove sufficiency.
Let~$\cC$ an $(n{=}2^r,M{=}2^{n-r})$ code such that
$\abs{\Ball_2(\bldc) \cap \cC} = 2$ for every codeword $\bldc \in \cC$.
We can then partition~$\cC$
(uniquely) into $M/2$ unordered pairs $\{ \bldc, \bldc' \}$
where $\distance(\bldc,\bldc') \le 2$.
We show that the capsules that correspond
to distinct pairs are disjoint.

Indeed, suppose that the capsules that correspond to
the pairs $\{ \bldc_1, \bldc_2 \}$
and $\{ \bldc_3, \bldc_4 \}$ intersect, i.e.,
there exists a word $\bldx \in \F_2^n$ in the intersection
\[
\left( \Ball_1(\bldc_1) \cup \Ball_1(\bldc_2) \right) \cap
\left( \Ball_1(\bldc_3) \cup \Ball_1(\bldc_4) \right)~.
\]
This means that $\bldx \in \Ball_1(\bldc_i) \cap \Ball_1(\bldc_j)$,
where $i \in \{ 1, 2 \}$ and $j \in \{ 3, 4 \}$.
By the triangle inequality,
\[
\distance(\bldc_i,\bldc_j)
\le \distance(\bldc_i,\bldx) + \distance(\bldc_j,\bldx) \le 2~,
\]
which means that $\bldc_i$ and~$\bldc_j$ are in the same capsule.
Yet this is possible
only if $\{ \bldc_1, \bldc_2 \} = \{ \bldc_3, \bldc_4 \}$.

Since the $M/2$ capsules are disjoint, the size of their union
is $(M/2)(2n) = 2^n$. Hence, they form a partition of $\F_2^n$,
and the result follows from Corollary~\ref{cor:nearly_part}.
\end{proof}

\section{Elementary Constructions of NP1CCs}
\label{sec:construct}

All the constructions of NP1CCs which will be presented in
this section are based on perfect codes and their properties.
Hence, we start this section by presenting some
basics of perfect codes. Recall that a perfect code is a code that meets
the bounds of~(\ref{eq:sphere-packing})
and~(\ref{eq:sphere-covering}).
We will consider only codes for which $R=1$ in these equations.
Such a code has length $n = 2^r -1$ and size $M  = 2^{n-r}$.
For each length $n$, there is an essentially unique linear perfect code
known as the Hamming code. A \Emph{perfect code}
can be a zeroed perfect code or its non-zeroed translate.
The number of nonequivalent perfect codes
is very large and it was considered throughout
the years~\cite{CHLL97}, \cite{Etz22}.
For example, it was proved in~\cite{EtVa94}, \cite{Phe84}, \cite{Vas62}
that the number of nonequivalent perfect codes of length $n$,
for sufficiently large~$n$ and a constant
$c=0.5 - \epsilon$, is at least $2^{2^{cn}}$.
Analysis of various constructions of such codes can be found
in~\cite[pp.~296--310]{CHLL97}.

An \Emph{extended zeroed perfect code} is obtained from
a zeroed perfect code by adding an even parity in a new coordinate.
There are two types of non-zeroed translates for
an extended zeroed perfect code, an odd translate and an even translate.
An \Emph{odd translate} of an extended zeroed perfect code contains
only words with odd weight including exactly one word of weight~$1$.
An \Emph{even translate} of an extended zeroed perfect code of
length~$2^r$ contains only words of even weight
including $2^{r-1}$ words of weight~$2$.
Since the zeroed perfect code is a perfect code
with covering radius~$1$, the following lemmas are followed.

\begin{lemma}
\label{lem:puncturePer}
If $\cC$ is an extended zeroed perfect code,
then deleting any one of its coordinates yields a perfect code.
\end{lemma}

\begin{lemma}
\label{lem:dist_x_ExHammE}
Let $\cC$ be an extended zeroed perfect code of length $n = 2^r$.
\begin{itemize}
\item[(1)]
For each word $\bldx \in \F_2^n$ of odd weight
there exists exactly one codeword $\bldc$ in $\cC$
such that $\distance(\bldc,\bldx)=1$.
\item[(2)]
For each word $x \in \F_2^n$ of even weight there exists exactly
one codeword~$\bldc$ in an odd translate of the extended
zeroed perfect code such that $\distance(\bldc,\bldx)=1$.
\item[(3)]
For each word $x \in \F_2^n$ of odd weight there exists exactly
one codeword $\bldc$ in an even translate of the extended
zeroed perfect code such that $\distance(\bldc,\bldx)=1$.
\end{itemize}
\end{lemma}

The simple construction in the next theorem yields NP1CCs
for all three types.

\begin{theorem}
\label{thm:Con_general}
Let $\cC_1$ and $\cC_2$ be
$(n{-}1{=}2^r{-}1,M{=}2^{n-r-1})$ perfect codes.
Then the code
\[
\cC \triangleq
\{ (\bldc,0) ~:~ \bldc \in \cC_1  \}
\;\cup\; \{ (\bldc,1) ~:~ \bldc \in \cC_2 \}
\]
is an $(n,M)$ NP1CC.
\end{theorem}

\begin{proof}
Since $\abs{\cC} = M = 2^{n-r}$, it suffices to show that
$\distance(\bldx,\cC) \le 1$ for every word $\bldx \in \F_2^n$.
Write $\bldx = (\bldy,b)$ where $b \in \F_2$.
Since $\cC_1$ and $\cC_2$ are perfect codes
we have $\distance(\bldy,\cC_1)\le 1$
and $\distance(\bldy,\cC_2)\le 1$;
hence, $\distance(\bldx,\cC) \le 1$ regardless of~$b$.\qed
\end{proof}

\begin{corollary}
\label{cor:con_three_types}
$~$
\begin{itemize}
\item[(1)]
If $\cC_1$ is a perfect code and
$\cC_1 = \cC_2$ in Theorem~\ref{thm:Con_general},
then the code $\cC$ is an NP1CC of \Type{\TA}.

\item[(2)]
If $\cC_1$ in Theorem~\ref{thm:Con_general} is a perfect code
and $\cC_2$ is a perfect code such that $\cC_1 \cap \cC_2 = \varnothing$, then
the code~$\cC$ is an NP1CC of \Type{\TB}.

\item[(3)]
If $\cC_1$ in Theorem~\ref{thm:Con_general} is a perfect code
and $\cC_2$ is a perfect code such that $\cC_1 \ne \cC_2$
and $\cC_1 \cap \cC_2 \ne \varnothing$, then the code $\cC$ is
an NP1CC of \Type{\TC}.
\end{itemize}
\end{corollary}
\begin{proof}
$~$
\begin{itemize}
\item[(1)]
This claim is immediate.

\item[(2)]
Since $\cC_1$ and $\cC_2$ are perfect codes and  $\cC_1 \cap \cC_2 = \varnothing$, it follows that for each
codeword $\bldc_1 \in \cC_1$ there exists a codeword $\bldc_2 \in \cC_2$ such that
$\distance(\bldc_1,\bldc_2)=1$ and therefore $\distance((\bldc_1,0),(\bldc_2,1))=2$.
This implies that $\cC$ is an NP1CC of \Type{\TB}.

\item[(3)]
For each $\bldc \in \cC_1 \cap \cC_2$ we have that $(\bldc,0),(\bldc,1) \in \cC$ and
hence $\distance((\bldc,0),(\bldc,1))=1$. Since $\cC_2$ is a perfect code, it follows that
for each $\bldc_1 \in \cC_1 \setminus \cC_2$,
there exists a codeword $\bldc_2 \in \cC_2 \setminus \cC_1$ such that
$\distance(\bldc_1,\bldc_2)=1$
and hence $\distance((\bldc_1,0),(\bldc_2,1))=2$.
This implies that $\cC$ is an NP1CC of \Type{\TC}.
\end{itemize}
\end{proof}

\begin{corollary}
\label{cor:number_inter}
If $\cC_1$ and $\cC_2$ in Theorem~\ref{thm:Con_general} are distinct zeroed
perfect codes and $\abs{\cC_1 \cap \cC_2} = k$, then the code $\cC$ is
an NP1CC of \Type{\TC} with exactly $k$ \Type{\Tone} pairs.
\end{corollary}

Corollaries~\ref{cor:con_three_types}(3) and~\ref{cor:number_inter}
raise an interesting question associated with $(n,M)$ NP1CCs
of \Type{\TC}. For which integer $k$, $1 \le k < M/2$,
there exists an NP1CC~$\cC$ with exactly $k$ pairs
of \Type{\Tone} and $M/2 - k$ pairs of \Type{\Ttwo}?
Corollary~\ref{cor:number_inter} implies that
such codes can be constructed from two zeroed perfect codes
whose intersection is $k$. It was proved by
Avgustinovich, Heden, and Solov'eva~\cite{AHS06}
that for each even integer $k$ such that
$0 \le k \le 2^{2^r - 2r}$
there exist two zeroed perfect codes of length $2^r-1$
whose intersection is~$k$. The minimum possible nonzero intersection
of two zeroed perfect codes is~$2$ and two such codes were found
in~\cite{EtVa98}. This intersection problem was
initiated in~\cite{EtVa94} and further investigated by
Avgustinovich, Heden, and Solov'eva~\cite{AHS05}.
A summary of the results with complete analysis were given by
Heden, Solov'eva, and Mogilnykh~\cite{HSM10}.

\begin{corollary}
There exist NP1CCs of \Type{\TA}, of \Type{\TB}, and of \Type{\TC}.
\end{corollary}

The next theorem provides a full characterization of NP1CCs
of \Type{\TA}.

\begin{theorem}
\label{thm:typeA_twoEx}
A code~$\cC$ is a zeroed $(n{=}2^r,M{=}2^{n-r})$ NP1CC of \Type{\TA},
if and only if it is the union of an extended zeroed perfect code of
length $n = 2^r$
with an odd translate of an extended zeroed perfect code of
the same length.
\end{theorem}

\begin{proof}
Suppose that~$\cC$ is a zeroed $(n{=}2^r,M{=}2^{n-r})$ NP1CC
of \Type{\TA}. Since its codewords can be partitioned into
\Type{\Tone} pairs, exactly half of the codewords have even weight.
Moreover, since there are no two codewords in~$\cC$
at distance~$2$ apart,
it follows that the sub-code that consists of the even-weight
(respectively, odd-weight) codewords has
minimum distance (at least)~$4$.
Therefore, the even-weight codewords in~$\cC$ form
an extended zeroed perfect code, and the odd-weight codewords
form an odd translate of an extended zeroed perfect code.

Conversely, suppose that $\cC = \cC_1 \cup \cC_2$,
where $\cC_1$ is an extended zeroed perfect code of length $n = 2^r$
and $\cC_2$ is
an odd translate of an extended zeroed perfect code of the same length.
If $\bldx \in \F_2^n$ is of even weight then,
by Lemma~\ref{lem:dist_x_ExHammE}(2),
there exists a codeword $\bldc \in \cC_2$
such that $\distance(\bldx,\bldc)=1$.
If $\bldx \in \F_2^n$ is of odd weight then,
by Lemma~\ref{lem:dist_x_ExHammE}(1),
there exists a codeword $\bldc \in \cC_1$
such that $\distance(\bldx,\bldc)=1$.
Moreover, $\abs{\cC_1 \cup \cC_2} = 2^{n-r}$ and,
hence, $\cC$ is a zeroed NP1CC of \Type{\TA}.
\end{proof}

\begin{corollary}
\label{cor:typeA_from_twoEx}
$~$
\begin{itemize}
\item[(1)]
A non-zeroed translate of a zeroed NP1CC of \Type{\TA}
is constructed as the union of an even translate of an
extended zeroed perfect code of length $2^r$ with an odd translate
of an extended zeroed perfect code of the same length.

\item[(2)]
The union of an even translate of an extended zeroed perfect code
of length $2^r$ with an odd translate of an extended
zeroed perfect code of the same length is a translate of
a zeroed NP1CC of \Type{\TA}.

\item[(3)]
There is a one-to-one correspondence between the pairs of
an extended zeroed perfect code and length~$2^r$ with
an odd translate of an extended zeroed perfect code of
the same length, and the zeroed NP1CCs of \Type{\TA}.

\item[(4)]
There is a one-to-one correspondence between the pairs of
an even translate of an extended zeroed perfect code
and length $2^r$ with an odd translate of an extended zeroed perfect
code of the same length, and the translates of zeroed NP1CCs of
\Type{\TA}.
\end{itemize}
\end{corollary}

Finally, other constructions in which an NP1CC of one type is obtained from
an NP1CC of another type will be given in Section~\ref{sec:equiv}.

\section{Weight Distribution of NP1CCs}
\label{sec:weightDist}

In this section, we characterize the weight distribution of NP1CCs.
In particular, we show that zeroed NP1CCs can have one out of
two weight distributions: one distribution is unique
to NP1CCs of \Type{\TA}, and the other is unique to NP1CCs
of \Type{\TB} (zeroed NP1CCs of \Type{\TC} can have any
of these two distributions).

Our analysis will make use of some known properties of
weight distributions, all of which can be found in
Chapters~5 and~6 in~\cite{McSl77}.
For the ease of reference,
we have summarized them in Section~\ref{sec:background}.

\subsection{Definitions and background}
\label{sec:background}

Given an $(n,M)$ code~$\cC$,
the \Emph{weight distribution} of~$\cC$ is the integer vector
$\bldA = \bldA_\cC = (A_i)_{i \in \Int{0:n}}$ with entries
\[
A_i = \abs{\cC \cap \Sphere_i}~.
\]
The respective \Emph{weight enumerator} is
the bivariate homogeneous polynomial
\[
\AA(x,y) = \sum_{i \in \Int{0:n}} A_i x^{n-i} y^i ~,
\]
or the univariate polynomial $\AA(y) \triangleq \AA(1,y)$.
The \Emph{distance distribution} of an $(n,M)$ code~$\cC$ is
the rational vector
$\bldB = \bldB_\cC = (B_i)_{i \in \Int{0:n}}$ whose entries are
\[
B_i =
\frac{1}{M}
\abs{\bigl\{ (\bldc,\bldc') \in \cC \times \cC
~:~ \distance(\bldc,\bldc') = i \bigr\}}~.
\]
Thus,
\begin{equation}
\label{eq:Bsum}
\bldB = \frac{1}{M} \sum_{\blde \in \cC} \bldA_{\blde + \cC}~.
\end{equation}
The respective \Emph{distance enumerator} is
the bivariate homogeneous polynomial
\[
\BB(x,y) = \sum_{i \in \Int{0:n}} B_i x^{n-i} y^i~,
\]
or the univariate polynomial $\BB(y) \triangleq \BB(1,y)$.

A zeroed code~$\cC$ is called \Emph{distance invariant}
if $\bldA_{\blde + \cC} = \bldA_\cC$
for every codeword $\blde \in \cC$.
For such codes we have $\bldB = \bldA$.
All linear codes are distance invariant.

Let $\bldz = (z_j)_{j \in \Int{n}}$
be a vector of real indeterminates and define the ring
\[
\Ring_n = \Realfield[\bldz] /
\langle z_1^2-1, z_2^2 - 1, \ldots, z_n^2 - 1 \rangle~.
\]
Namely, the elements and arithmetic in $\Ring_n$ are obtained from
those in $\Realfield[\bldz]$ by reducing modulo~$2$
the exponents of powers of the indeterminates (and so those powers
can be seen as the elements~$0$ and~$1$ of $\F_2$).
For $\bldu = (u_j)_{j \in \Int{n}} \in \F_2^n$,
we introduce the shorthand notation
\[
\bldz^\bldu = \prod_{j \in \Int{n}} z_j^{u_j}~.
\]

For each $\bldu = (u_j)_{j \in \Int{n}} \in \F_2^n$,
we define the \Emph{character}
$\chi_\bldu : \Ring_n \rightarrow \Realfield$
which maps any
\[
\GG = \GG(\bldz)
= \sum_{\bldv \in \F_2^n} g_\bldv \bldz^\bldv \in \Ring_n
\]
to its value at $\bldz = \left((-1)^{u_j}\right)_{j \in \Int{n}}$:
\[
\chi_\bldu(\GG(\bldz))
= \sum_{\bldv \in \F_2^n}
g_\bldv \cdot (-1)^{\langle \bldu,\bldv \rangle}~,
\]
where $\langle \cdot,\cdot \rangle$ denotes dot product.
Clearly, $\chi_\bldu$ is linear over~$\Realfield$
and multiplicative.

With each $(n,M)$ code~$\cC$ we associate
its \Emph{generating function} in $\Ring_n$:
\[
\Code(\bldz) = \sum_{\bldu \in \cC} \bldz^\bldu~.
\]

Given an $(n,M)$ code~$\cC$,
the \Emph{transform} of the weight distribution $\bldA_\cC$ is
the rational vector
$\bldA' = \bldA'_\cC = (A'_i)_{i \in \Int{0:n}}$ with the entries
\begin{equation}
\label{eq:A'}
A'_i = \frac{1}{M} \sum_{\bldu \in \Sphere_i}
\chi_\bldu(\Code(\bldz))~.
\end{equation}
In particular, $A'_0 \equiv 1$.
The respective enumerator polynomial,
\[
\AA'(x,y) = \sum_{i \in \Int{0:n}} A'_i x^{n-i} y^i~,
\]
is related to $\AA(x,y)$ by \Emph{MacWilliams' identities}:
\begin{equation}
\label{eq:MacWilliams1}
\AA'(x,y) = \frac{1}{M} \cdot \AA(x+y,x-y) \phantom{.}
\end{equation}
and
\begin{equation}
\label{eq:MacWilliams2}
\AA(x,y) = \frac{M}{2^n} \cdot \AA'(x+y,x-y)~.
\end{equation}
When~$\cC$ is linear, the transform $\bldA'$ is
the weight distribution of the dual code, $\cC^\perp$, of~$\cC$.

\begin{example}
\label{ex:linearNP1CCcontinued}
Let $\cC_1$ be the \Type{\TA} linear NP1CC in
Example~\ref{ex:linearNP1CC}.
The dual code $\cC_1^\perp$ is the simplex code padded with
an extra zero coordinate; hence,
\[
\AA'_{\cC_1}(x,y)
= x^n + (n-1) \, x^{n/2} y^{n/2}~.
\]
The weight enumerator of $\cC_1$ is therefore
\begin{eqnarray}
\AA_1(y) \triangleq \AA_{\cC_1}(y)
& = &
\frac{1}{n} (1+y)^n
+ \left( 1 - \frac{1}{n} \right) (1+y)^{n/2} (1-y)^{n/2}
\nonumber \\
& = &
\label{eq:TypeA}
\frac{1}{n} (1+y)^n
+ \left( 1 - \frac{1}{n} \right) (1-y^2)^{n/2} ~.
\end{eqnarray}
Let $\cC_2$ be the \Type{\TB} linear NP1CC in that example.
The dual code $\cC_2^\perp$ is the simplex code padded with a replica of
one of the coordinates. Here
\[
\AA'_{\cC_2}(x,y) =
x^n + \left( \frac{n}{2} - 1 \right) x^{n/2} y^{n/2}
+ \frac{n}{2} \, x^{n/2-1} y^{n/2+1}
\]
and, so,
\begin{eqnarray}
\AA_2(y) \triangleq \AA_{\cC_2}(y)
& = &
\frac{1}{n} (1+y)^n
+
\left( \frac{1}{2} - \frac{1}{n} \right) (1+y)^{n/2} (1-y)^{n/2}
\nonumber \\
\label{eq:TypeB}
&&
\quad \quad {} + \frac{1}{2} (1+y)^{n/2-1} (1-y)^{n/2+1}~.
\end{eqnarray}%
\qed
\end{example}

The transform of the distance distribution~$\bldB$ is
the rational vector $\bldB' = (B'_i)_{i \in \Int{0:n}}$ with the entries
\begin{equation}
\label{eq:B'}
B'_i = \frac{1}{M^2} \sum_{u \in \Sphere_i}
\bigl( \chi_\bldu(\Code(\bldz)) \bigr)^2~.
\end{equation}
The respective enumerator polynomial,
\[
\BB'(x,y) = \sum_{i \in \Int{0:n}} B'_i x^{n-i} y^i~,
\]
is related to $\BB(x,y)$ by
MacWilliams'
identities~(\ref{eq:MacWilliams1})--(\ref{eq:MacWilliams2}),
with $\AA(x,y)$ and $\AA'(x,y)$ therein replaced by
$\BB(x,y)$ and $\BB'(x,y)$.
When a zeroed code~$\cC$ is distance invariant
we have $\bldB' = \bldA'$.

By~(\ref{eq:B'}) it follows that
\begin{equation}
\label{eq:B'support}
B'_i = 0 \;\Longleftrightarrow\;
\chi_\bldu(\Code(\bldz)) = 0
\; \textrm{for all} \; \bldu \in \Sphere_i~.
\end{equation}
Hence, by~(\ref{eq:A'}),
\begin{equation}
\label{eq:A'support}
\Support(\bldA') \subseteq \Support(\bldB')~.
\end{equation}

The \Emph{external distance} of~$\cC$ is defined by
\[
s' = \abs{\Support(\bldB') \setminus \{ 0 \}}
= \weight(\bldB') - 1~.
\]

\begin{theorem}[{\cite[Ch.~6, Thm.~20]{McSl77}}]
\label{thm:uniqueness}
Let~$\cC$ be an $(n,M)$ code with external distance~$s'$.
Then for any $\blde \in \F_2^n$,
the entries of $\bldA_{\blde + \cC}$ are uniquely determined by
$n$, $M$, $\Support(\bldB')$,
and the first $s'$ entries of $\bldA_{\blde + \cC}$.
\end{theorem}

It follows from (the proof of) this theorem that
a code is distance invariant whenever
its external distance does not exceed its minimum distance.
Moreover, the external distance bounds from above
the covering radius of the code.

\subsection{Characterization of the weight distribution of NP1CCs}
\label{sec:characterization}

Our next theorem will be the main tool for characterizing
the weight distribution of NP1CCs. Our proof
will use the following notation.
For $i \in \Int{0:n}$, we let $\YY_i(\bldz)$ be
the $i$th
\Emph{elementary symmetric function} in the entries of~$\bldz$:
\[
\YY_i(\bldz) = \sum_{\bldu \in \Sphere_i}
\bldz^\bldu~.
\]
It is known (see~\cite[p.~135]{McSl77})
that for any $\bldu \in \Sphere_w$,
\begin{equation}
\label{eq:chiY}
\chi_\bldu(\YY_i(\bldz)) = P_i(w)~,
\end{equation}
where $P_i(\cdot)$ is the $i$th \Emph{Krawtchouk polynomial}:
\[
P_i(w)
\triangleq
\sum_{j \in \Int{0:i}} (-1)^j \binom{w}{j} \binom{n-w}{i-j}~.
\]

\begin{theorem}
\label{thm:externaldistance}
Let~$\cC$ be an $(n,M)$~NP1CC and let~$\bldB'$ be
the transform of its distance distribution.
Then
\[
\Support(\bldB') \subseteq \{ 0, n/2, n/2 + 1 \}~,
\]
i.e., $s' \le 2$.
\end{theorem}

\begin{proof}
Let $\Code(\bldz)$ be the generating function of $\cC$
and consider the following multinomial (in $\Ring_n$):
\[
\Code(\bldz)
\cdot \sum_{\bldu \in \Sphere_1 \cup \Sphere_2} \bldz^\bldu
= \Code(\bldz) \left( \YY_1(\bldz) + \YY_2(\bldz) \right) ~.
\]
For any word $\bldx \in \F_2^n$,
the coefficient of $\bldz^\bldx$ in this multinomial
equals the number of codewords at distance~$1$ or~$2$ from~$\bldx$.
By Theorems~\ref{thm:radius2} and~\ref{thm:adjacent_codewords},
this number is
\[
\renewcommand{\arraystretch}{1.3}
\left\{
\begin{array}{ccl}
\displaystyle\frac{n}{2} + 1
&& \textrm{if $\bldx$ is a non-codeword,} \\
1
&& \textrm{if $\bldx$ is a codeword.}
\end{array}
\right.
\]
Hence,
\begin{eqnarray*}
\Code(\bldz) \left( \frac{n}{2} + \YY_1(\bldz) + \YY_2(\bldz) \right)
& =  &
\left( \frac{n}{2} + 1 \right) \sum_{\bldu \in \F_2^n} \bldz^\bldu \\
& = &
\left( \frac{n}{2} + 1 \right) \prod_{j \in \Int{n}} (1 + z_j)^n
\end{eqnarray*}
and, so, for every $\bldu \in \F_2^n \setminus \{ \bldzero \}$,
\[
\chi_\bldu
\left( \Code(\bldz)
\left( \frac{n}{2} + \YY_1(\bldz) + \YY_2(\bldz) \right)
\right) = 0~.
\]
By~(\ref{eq:chiY}) and the multiplicativity of~$\chi_\bldu(\cdot)$
we get
\begin{equation}
\label{eq:null}
\chi_\bldu(\Code(\bldz)) \cdot \beta(\weight(\bldu)) = 0~,
\end{equation}
where $\beta(\cdot)$ is the following polynomial:
\begin{eqnarray*}
\beta(w)
& = &
\frac{n}{2} + P_1(w) + P_2(w) \\
& = &
\frac{n}{2} + (n - 2 w) + \left( \binom{n}{2} - 2 n w + 2 w^2 \right) \\
& = &
2 \left( w - \frac{n}{2} \right) \left( w - \frac{n}{2} - 1 \right)~.
\end{eqnarray*}

Let~$w$ be a nonzero element in $\Support(\bldB')$,
namely, $B_w' \ne 0$.
By~(\ref{eq:B'support}), there exists at least one word
$\bldu \in \Sphere_w$ such that
$\chi_\bldu(\Code(\bldz)) \ne 0$. Hence, by~(\ref{eq:null}),
\[
\beta(w) = 0
\]
(see Lemma~19 in~\cite[Ch.~6]{McSl77}),
i.e., $w \in \{ n/2, n/2 + 1 \}$.
\end{proof}

Let $\cC$ be an $(n,M)$~NP1CC which, without any loss of generality,
we assume to be zeroed, and let $\blde + \cC$ be any of its translates.
By Theorem~\ref{thm:externaldistance} we have $s' \le 2$
and, so, by Theorem~\ref{thm:uniqueness},
the weight distribution, $\bldA = (A_i)_{i \in \Int{0:n}}$,
of $\blde + \cC$ is uniquely determined by its first two entries,
namely, by the pair $(A_0 \; A_1)$.
And by Corollary~\ref{cor:cover_nonC} and Theorems~\ref{thm:radius2},
this pair can take (only) four values,
as shown in the first three column in Table~\ref{tab:cases}.
\begin{table}[hbt]
\caption{Parameters of the four possible weight distributions
of NP1CCs.}
\label{tab:cases}
\normalsize
\[
\begin{array}{c@{\quad\quad}cc@{\quad\quad}cc@{\quad\quad}c}
\hline\hline
\textrm{Case} & A_0 & A_1 & A'_{n/2} & A'_{n/2+1} & \textrm{Types} \\
\hline
\blde \in \cC \; \textrm{and} \; \abs{\Ball_1(\blde) \cap \cC} = 2
& 1 & 1 & n-1 & 0 & \TA, \TC \\
\blde \in \cC \; \textrm{and} \; \abs{\Ball_1(\blde) \cap \cC} = 1
& 1 & 0 & n/2 - 1 & n/2 & \TB, \TC \\
\blde \not\in \cC \; \textrm{and} \; \abs{\Ball_1(\blde) \cap \cC} = 2
& 0 & 2 & n/2 - 1 & -n/2  & \TB, \TC \\
\blde \not\in \cC \; \textrm{and} \; \abs{\Ball_1(\blde) \cap \cC} = 1
& 0 & 1 & - 1 & 0 & \TA, \TB, \TC \\
\hline\hline
\end{array}
\]
\end{table}
In what follows, we compute the explicit dependence of
the weight enumerator $\AA(y)$ (and, hence, of
the weight distribution~$\bldA$) on $(A_0 \; A_1)$.
We do this by first determining the transform $\AA'(x,y)$
using the first set of MacWilliams' identities~(\ref{eq:MacWilliams1});
then, we use the second set~(\ref{eq:MacWilliams2})
to obtain the complete weight enumerator $\AA(x,y)$.

Substituting $(x,y) = (1,1)$ in both sides
of~(\ref{eq:MacWilliams1})
and recalling that $A'_0 \equiv 1$ and (from~(\ref{eq:A'support})
and Theorem~\ref{thm:externaldistance})
that $\Support(\bldA') \subseteq \Support(\bldB')
\subseteq \{ 0, n/2, n/2 + 1 \}$, we get
\[
1 + A'_{n/2} + A'_{n/2 + 1} = n~.
\]
Next, differentiating both sides of~(\ref{eq:MacWilliams1})
with respect to~$y$ and doing the same substitution yields
\[
\frac{n}{2} \, A'_{n/2}
+ \left( \frac{n}{2} + 1 \right) A'_{n/2 + 1} =
\frac{n}{2} (n A_0 - A_1)~.
\]
Solving the last two equations for $A'_{n/2}$
and $A'_{n/2+1}$ in terms of $(A_0 \; A_1)$ results in:
\begin{equation}
\label{eq:solution}
\renewcommand{\arraystretch}{2}
\begin{array}{lcl}
A'_{n/2}   & = & \displaystyle n A_0 - \frac{n}{2} (1 - A_1) - 1 \\
A'_{n/2+1} & = & \displaystyle \frac{n}{2} (1 - A_1)~.
\end{array}
\end{equation}
The fourth and fifth columns in Table~\ref{tab:cases}
present the solutions for $A'_{n/2}$ and $A'_{n/2+1}$
(and, thus, the complete characterization of the transform $\AA'(x,y)$)
for each of the four cases in the table.
Knowing now all the nonzero coefficients in $\AA'(x,y)$,
we get from~(\ref{eq:MacWilliams2}) the complete weight enumerator
$\AA(y)$, in terms of $(A_0 \; A_1)$:
\begin{eqnarray*}
\AA(y)
& = &
\frac{1}{n} (1+y)^n
+ \left( A_0 - \frac{1 - A_1}{2} - \frac{1}{n} \right)
(1+y)^{n/2} (1-y)^{n/2} \\
&&
\quad \quad {} + \frac{1-A_1}{2} \cdot (1+y)^{n/2-1} (1-y)^{n/2+1}~.
\end{eqnarray*}
Rearranging terms leads to the following result.

\begin{theorem}
\label{thm:weightenumerator}
Let $\cC$ be a zeroed $(n,M)$~NP1CC and
let $\blde$ be a word in $\F_2^n$.
Then the weight enumerator of $\blde + \cC$ is given by
\begin{equation}
\label{eq:weightenumerator}
\AA(y) = \frac{1}{n} (1+y)^n
+ \left( A_0 - \frac{1}{n}
+ \Bigl( A_0 + A_1 - 1 - \frac{1}{n} \Bigr) y \right)
(1-y) (1-y^2)^{n/2-1}~,
\end{equation}
where $(A_0 \; A_1)$ is determined from $\cC$ and $\blde$
according to Table~\ref{tab:cases}.
\end{theorem}

We next present an explicit expression for
the entries of the weight distribution
$\bldA = (A_i)_{i \in \Int{0:n}}$. For $i \in \Int{0:n}$, let
\[
\Delta_i \triangleq
(-1)^{\lceil i/2 \rceil}
\binom{n/2-1}{\lfloor i/2 \rfloor}
\]
(where the binomial coefficient is assumed to be zero for
invalid parameters); it can be verified that
\[
(1-y) (1-y^2)^{n/2-1} = \sum_{i \in \Int{0:n}} \Delta_i y^i ~.
\]
By~(\ref{eq:weightenumerator}) it then follows that for
every $i \in \Int{0:n}$,
\[
A_i =
\frac{1}{n} \binom{n}{i}
+ \Bigl( A_0 - \frac{1}{n} \Bigr) \Delta_i
+ \Bigl( A_0 + A_1 - 1 - \frac{1}{n} \Bigr) \Delta_{i-1} ~.
\]

When $(A_0 \; A_1) = (1 \; 1)$,
Eq.~(\ref{eq:weightenumerator}) becomes
$\AA_1(y)$ in~(\ref{eq:TypeA}).
Note that this case can occur only when $\cC$ is
either of \Type{\TA} or of \Type{\TC}
(see the last column in Table~\ref{tab:cases}).
Moreover, if $\cC$ is of \Type{\TA}, then
$\AA_1(y)$ is the weight enumerator of $\blde + \cC$
for \Emph{every} codeword $\blde \in \cC$.
Hence, \Type{\TA} codes are distance invariant:
in their case $\bldB = \bldA$ and $\bldB' = \bldA'$
and, consequently, their external distance is~$1$
(which is also their minimum distance).

When $(A_0 \; A_1) = (1 \; 0)$,
Eq.~(\ref{eq:weightenumerator}) becomes
$\AA_2(y)$ in~(\ref{eq:TypeB}).
This case can occur only when $\cC$ is
either of \Type{\TB} or of \Type{\TC}.
By a similar reasoning as before we conclude that
\Type{\TB} codes are distance invariant as well
and their external distance, as well as their minimum distance, is~$2$
(except when $n = 2$, where the external distance is~$1$).

The case $(A_0 \; A_1) = (0 \; 2)$ also pertains
to \Type{\TB} and \Type{\TC} codes,
as it occurs when $\blde$ is a midword.
Eq.~(\ref{eq:weightenumerator}) is then similar to~(\ref{eq:TypeB})
except that the sign of the last term in~(\ref{eq:TypeB}) is flipped.

Finally, the case $(A_0 \; A_1) = (0 \; 1)$
corresponds to $\blde$
being a non-codeword that is not a midword.
This case can occur in all types,
and the weight enumerator is
\[
\frac{1}{n} \left( (1+y)^n - (1-y^2)^{n/2} \right)~.
\]

\Type{\TC} codes cannot be distance invariant, since
a fraction $B_1 \in (0,1)$ of the codewords
have $1$-neighbors while the other codewords do not.
Still, by~(\ref{eq:Bsum}), we get a complete characterization
of their distance enumerator:
\[
\BB(y) = B_1 \cdot \AA_1(y) + (1-B_1) \cdot \AA_2(y)~.
\]

\begin{corollary}
\label{cor:evenweight}
Let $\cC$ be an $(n,M)$ N1PCC where $n > 2$.
Then exactly half of the codewords in $\cC$ have even weight.
\end{corollary}

\begin{proof}
It follows from~(\ref{eq:weightenumerator}) that
\[
\sum_{i \textrm{ even}} A_i - \sum_{i \textrm{ odd}} A_i
= \AA(-1) = 0 ~.
\]
\end{proof}

\begin{corollary}
\label{cor:evenpairs}
Let $\cC$ be an $(n,M)$ N1PCC where $n > 2$.
Then the number, $k$, of \Type{\Tone} pairs in $\cC$ is even
(and so is the number, $M/2 - k$, of \Type{\Ttwo} pairs).
Moreover, exactly half of the \Type{\Ttwo} pairs consist of
even-weight partners.
\end{corollary}

\begin{proof}
Within each \Type{\Tone} pair, one (and only one) of
the partners has even weight. Hence,
in the subset $\cC_\Tone$ of $\cC$ formed by the union
of all \Type{\Tone} pairs,
exactly half the codewords have even weight.
By Corollary~\ref{cor:evenweight} it then follows that
the same must hold in the subset $\cC_\Ttwo = \cC \setminus \cC_\Tone$,
which is formed by the union of all \Type{\Ttwo} pairs.
Yet in each \Type{\Ttwo} pair,
the parity of the partners must be the same;
hence, there are as many \Type{\Ttwo} pairs with
even-weight partners as such pairs with odd-weight partners.
We conclude that $\abs{\cC_\Ttwo}$ is even and, therefore,
so is $k = \abs{\cC_\Tone} = M/2 - \abs{\cC_\Ttwo}$.
\end{proof}

\begin{remark}
The weight distributions of \Type{\TA} and \Type{\TB} NP1CCs
were shown in~\cite{BoEt24} using a different technique.
Another method for computing the weight distributions of
the three types was suggested by the reviewer and is based
on equitable partitions and
quotient matrices~\cite{Kro11}, \cite{Mar92}.
This method completely solves the weight distribution for
\Type{\TA} and \Type{\TB}.
For \Type{\TC}, we need to consider the same technique for
the extended code and analyze its punctured code after
the solution of the weight distribution. However this method does not
recover any information on the distance distribution.
\end{remark}

\section{Balanced Nearly Perfect Covering Codes}
\label{sec:balanced}

There are many NP1CCs which have some additional special properties.
One example of such property is a code of \Type{\TA} in which
for each coordinate there is at least one pair of partners
that disagree on that coordinate
(such a property will turn out to be useful in Section~\ref{sec:equiv}).
In this section, we construct such codes.
Moreover, for the constructed code, for any
given coordinate, the number of \Type{\Tone} pairs
that contain partners
that disagree on the given coordinate is $2^{2^r-2r-1}$.
In other words, this number is the same for all coordinates.
Such a code will be called a \Emph{balanced NP1CC}
and it can be constructed recursively, as we show below.

A \Emph{self-dual} sequence is a binary cyclic sequence
that is equal to its complement.
If there is no periodicity in the sequence, then it can be written
as $[X ~ \bar{X}]$, where $\bar{X}$ is the binary complement of~$X$.
The following two cyclic sequences $\dS_1=[00011011 ~ 11100100]$
and $\dS_2=[00011010 ~ 11100101]$ are self-dual sequences of length~$16$.
We consider all the $32$~words obtained by any eight consecutive symbols of $\dS_1$ and $\dS_2$.
In these $32$~words, we have $16$ even-weight words of
length~$8$ and $16$~odd-weight words of length~$8$.
Let $\cC$ be the code obtained from these $32$~words.
Let $\cC_e$ be the code obtained from the $16$~even-weight words
of $\cC$ and $\cC_o$ be code obtained from the $16$~odd-weight words
of $\cC$. The code~$\cC_e$ is an even translate of
an extended zeroed perfect code of length~$8$ and $\cC_o$ is
an odd translate of an extended zeroed perfect code
of length~$8$. Therefore, by Corollary~\ref{cor:typeA_from_twoEx}(2) their union is a non-zeroed translate of an NP1CC of \Type{\TA}.
Finally, for each one of the eight coordinates, there are exactly
two \Type{\Tone} pairs from $\cC_e$ and $\cC_o$,
where the partners in each pair disagree exactly
on this coordinate, and hence the code is balanced.
To obtain a zeroed NP1CC from this code we have
to translate it by one of its codewords.

\begin{example}
Three more pairs of sequences can be used as $\dS_1$ and $\dS_2$
(each pair have disjoint codewords of length 8 and
each pair can be obtained from each other by decimation)
\[
\dS_1 = [01001111 ~ 10110000]~, ~~~ \dS_2 = [01001110 ~ 10110001]~,
\]
\[
\dS_1 = [01110111 ~ 10001000]~, ~~~ \dS_2 = [01110110 ~ 10001001]~,
\]
\[
\dS_1 = [00100010 ~ 11011101]~, ~~~ \dS_2 = [00100011 ~ 11011100]~.
\]
\qed
\end{example}

Generally, we consider $2^{2^{r-1} -2r+1}$ self-dual sequences of
length $2^r$. Let $\cC$ be the set of
$2^{2^{r-1}-r+1}$ words obtained by any
$2^{r-1}$ consecutive symbols in these self-dual sequences.
Assume further that all these $2^{2^{r-1}-r+1}$ words of
length $2^{r-1}$ are different.
Let $\cC_e$ be the set of even-weight words
in $\cC$ and $\cC_o$ be the set of odd-weight words in $\cC$.
Assume further that $\cC_e$ and $\cC_o$ are
two translates of extended zeroed perfect codes of length $2^{r-1}$
(one even translate and one odd translate).
Assume further that the $2^{2^{r-1} -2r+1}$ self-dual sequences can be
ordered in pairs
\[
\cP_i =([X ~ \bar{X}],[X' ~ \bar{X}']), ~~ 1 \le i \le 2^{2^{r-1} -2r}
~,
\]
where $X$ and~$X'$ are sequences of length $2^{r-1}$ which start
with a `$0$' and differ only in their last symbol.

This partition into pairs of self-dual sequences implies that
the codewords of $\cC_e$ and~$\cC_o$ can be partitioned into pairs
of codewords defined by
the following set (see also the proof of Lemma~\ref{lem:balanced_code}).
\[
\cQ \triangleq \{ \{\bldc_1,\bldc_2\} ~:~ \bldc_1 \in \cC_e ,
~ \bldc_2 \in \cC_o, ~ \distance(\bldc_1,\bldc_2)=1 \} ~,
\]
where $\cQ$ contains exactly $2^{2^{r-1} -r}$ pairs of codewords and
each codeword of $\cC_e$ and each codeword of $\cC_o$ is contained
in exactly one such pair. Such a definition for~$\cQ$ and
the definition of the pairs in $\cP_i$, $1 \le i \le 2^{2^{r-1} -2r}$,
imply that for each one of the $2^{r-1}$ coordinates,
there are exactly $2^{2^{r-1}-2r+1}$ pairs which contain
codewords that disagree only at this coordinate.

For each pair of self-dual sequences
$\cP_i =([X ~ \bar{X}],[X' ~ \bar{X}'])$, $1 \le i \le 2^{2^{r-1} -2r}$,
and any sequence $V=(0,Z)$ of length $2^r$, where $Z$~is
an even-weight sequence of length $2^r-1$, we form the following pair
\[
\cP_{iV}
= ([ V ~ X+V ~ \bar{V} ~ X+\bar{V}],
[ V ~ X'+V ~ \bar{V} ~ X'+\bar{V}] ).
\]

The following lemma is an immediate observation.
\begin{lemma}
The two sequences in $\cP_{iY}$ are self-dual sequences.
They have the form $[X_1 ~ X_2 ~ \bar{X}_1 ~\bar{X}_2 ]$
and $[X_1 ~ X'_2 ~ \bar{X}_1 ~\bar{X}'_2 ]$, where $X_1$
and $X_2$ are words of length $2^{r-1}$ that start with a `$0$'.
\end{lemma}

Let $\cE$ be the set of even-weight sequences of length $2^{r-1}$
that start with a `$0$'.
Let $\cC'$ be the code defined by taking the union of all
the sequences in these pairs and
from each sequence taking $2^{r+1}$ codewords obtained from
the consecutive $2^r$ bits of the sequence starting from each of
the $2^{r+1}$ entries of the sequence.

The construction for the pairs of sequences is very similar to
the constructions presented
in~\cite{Etz24}, \cite{EtLe84}, \cite{EtPa96}.
The same code was defined and analyzed for another purpose
in~\cite{CETV24}.
The following observations lead to the main result.
The first lemma was proved
in~\cite{Etz24}, \cite{EtLe84}, \cite{EtPa96}.

\begin{lemma}
\label{lem:wordsCdistinct}
All the words of length $2^r$ obtained from all
the pairs $\cP_{iV}$, $1 \le i \le 2^{2^{r-1} -2r}$, $V \in \cE$
are distinct.
\end{lemma}

\begin{corollary}
\label{cor:Cnum_codew}
The code $\cC'$ contains $2^{2^r -r}$ codewords.
\end{corollary}

The following lemma was mentioned in~\cite{CETV24} without a proof.

\begin{lemma}
\label{lem:C_covering1}
The code $\cC'$ is an NP1CC.
\end{lemma}

\begin{proof}
The form of the two sequences in a pair implies that we can partition
the $2^{2^r -r}$ codewords of $\cC'$ into two sets,
one with words of even weight and one with words of odd weight.
We claim that there are no two codewords at distance~$2$ apart.
Assume to the contrary that there are two such distinct codewords,
$(X_1 , X_2)$ and $(Y_1 , Y_2)$ where $X_1,X_2,Y_1,Y_2$
are sequences of length~$2^r$ and
${\distance((X_1 , X_2),(Y_1 , Y_2))=2}$.
The associated two self-dual sequences (not necessarily distinct)
of length~$2^{r+1}$ are
\[
[ X_1 ~ X_2 ~ \bar{X}_1 ~ \bar{X}_2 ]
~~ \text{and} ~~ [ Y_1 ~ Y_2 ~ \bar{Y}_1 ~ \bar{Y}_2 ]~.
\]
We distinguish now between two cases:

\Emph{Case 1:}
$\distance(X_1,Y_1)=2$ and $X_2=Y_2$
(the case $\distance(X_2,Y_2)=2$ and $X_1=Y_1$ is equivalent).
The code $\cC$ contains the codewords $X_1 + X_2$ and $Y_1 + Y_2$,
where $\distance(X_1+X_2,Y_1+Y_2)=2$, a contradiction.

\Emph{Case 2:}
$\distance(X_1,Y_1)=1$ and $\distance(X_2,Y_2)=1$.
The code $\cC$ contains the codewords $X_1 + X_2$ and $Y_1 + Y_2$,
where either $\distance(X_1+X_2,Y_1+Y_2)=2$
or $\distance(X_1+X_2,Y_1+Y_2)=0$.
It is not possible to have
$\distance(X_1+X_2,Y_1+Y_2)=2$ since
the code $\cC$ does not contains two codewords at distance~$2$ apart.
If $\distance(X_1+X_2,Y_1+Y_2)=0$, then the coordinate on which
$X_1$ and $Y_1$ differ is the same coordinate where
$X_2$ and $Y_2$ differ.
This implies that the two distinct self-dual sequences
\begin{equation}
\label{eq:twoSDforCont}
[ X_1 ~ X_2 ~ \bar{X}_1 ~ \bar{X}_2 ]
~~ \text{and} ~~ [ Y_1 ~ Y_2 ~ \bar{Y}_1 ~ \bar{Y}_2 ]~
\end{equation}
are obtained from the same self-dual sequences
$[X_1 + X_2 ~ \bar{X}_1 + X_2] = [Y_1 + Y_2 ~ \bar{Y}_1 + Y_2]$.
The two sequences in~(\ref{eq:twoSDforCont}) differ in four positions,
each two are separated by $2^{r-1}-1$ equal positions.
But our choice of $V=(0,Z)$ of length $2^r$,
where $Z$~has even weight,
cannot yield two sequences that differ in exactly one position
among $2^r$ consecutive coordinates, thereby resulting
in a contradiction.

Hence, the minimum distance in each set of codewords is~$4$,
which implies that each set of words has the parameters
of the extended zeroed perfect code. Thus, $\cC'$ is an NP1CC.
\end{proof}

\begin{lemma}
\label{lem:balanced_code}
The code $\cC'$ is a balanced NP1CC.
\end{lemma}

\begin{proof}
By Corollary~\ref{cor:Cnum_codew} and Lemma~\ref{lem:C_covering1}
we have that $\cC'$~is an NP1CC.
Two pairs of sequences differ in positions $2^r$ and $2^{r+1}$.
These two positions are associated with the last coordinate
of the codewords that start in the first bit and bit $2^r+1$
of these two sequences.
Since the codewords are formed from the $2^r$ consecutive bits
in each pair of such sequences, the codewords
which start in the next bits differ in the previous positions and so on.
It follows that for each position~$\gamma$
there are exactly two pairs of codewords from these two sequences
which differ exactly in position~$\gamma$.
Therefore, $\cC'$ is a balanced NP1CC.
\end{proof}

\begin{example}
For a code of length~$8$ there is one pair of self-dual sequences of
length~$16$ given by
\[
\cP = ( [00011011 ~ 11100100],[00011010 ~ 11100101]) ~.
\]
Applying the recursion we obtain the following $64$~pairs
(the first eight and the last four are given),
where the index is their place in the lexicographic order
and the first eight bits are ordered by this lexicographic order
\begin{align*}
 \cP_1 = ([00000000 ~ 00011011 ~ 11111111 ~ 11100100],[00000000 ~ 00011010 ~ 11111111 ~ 11100101])  \\
 \cP_2 = ([00000011 ~ 00011000 ~ 11111100 ~ 11100111],[00000011 ~ 00011001 ~ 11111100 ~ 11100110])  \\
 \cP_3 = ([00000101 ~ 00011110 ~ 11111010 ~ 11100001],[00000101 ~ 00011111 ~ 11111010 ~ 11100000])  \\
 \cP_4 = ([00000110 ~ 00011101 ~ 11111001 ~ 11100010],[00000110 ~ 00011100 ~ 11111001 ~ 11100011])  \\
 \cP_5 = ([00001001 ~ 00010010 ~ 11110110 ~ 11101101],[00001001 ~ 00010011 ~ 11110110 ~ 11101100])  \\
 \cP_6 = ([00001010 ~ 00010001 ~ 11110101 ~ 11101110],[00001010 ~ 00010000 ~ 11110101 ~ 11101111])  \\
 \cP_7 = ([00001100 ~ 00010111 ~ 11110011 ~ 11101000],[00001100 ~ 00010110 ~ 11110011 ~ 11101001])  \\
 \cP_8 = ([00001111 ~ 00010100 ~ 11110000 ~ 11101011],[00001111 ~ 00010101 ~ 11110000 ~ 11101010])  \\
 \vdots ~~~~~~~~~~~~~~~~~~~~~~~~~~~~~~~~~~~~~~~~~~~~~~~~~~~~~~~~~~~~~~~ \\
 \cP_{61} = ([01110111 ~ 01101100 ~ 10001000 ~ 10010011],[01110111 ~ 01101101 ~ 10001000 ~ 10010010])  \\
 \cP_{62} = ([01111011 ~ 01100000 ~ 10000100 ~ 10011111],[01111011 ~ 01100001 ~ 10000100 ~ 10011110])  \\
 \cP_{63} = ([01111101 ~ 01100110 ~ 10000010 ~ 10011001],[01111101 ~ 01100111 ~ 10000010 ~ 10011000])  \\
 \cP_{64} = ([01111110 ~ 01100101 ~ 10000001 ~ 10011010],[01111110 ~ 01100100 ~ 10000001 ~ 10011011])  \\
\end{align*}
\qed
\end{example}

\section{Extended NP1CCs and their properties}
\label{sec:equiv}

In this section, we show how to construct one type of NP1CCs
from another type in a rather straightforward way.
We also show that we can partition
the codes of the three types into some logical equivalence classes,
where each equivalence class can contain
NP1CCs from more than one type, i.e., from two of them or even from
all the three types.
This will be done by considering the extended codes of NP1CCs.

Given an $(n{=}2^r,M{=}2^{n-r})$ NP1CC $\cC$,
we construct its extended code $\cC^*$ of length $n+1$
by adding an even parity to each one of its codewords.
Such an extended NP1CC will be called \Emph{ENP1CC}.
The following property is an immediate consequence from the definitions.

\begin{lemma}
\label{lem:part_d=2}
In an ENP1CC $\cC^*$,
each codeword has even weight and for each codeword
$c \in \cC^*$ there exists exactly one codeword $\bldc' \in \cC^*$
such that $\distance(\bldc,\bldc')=2$.
For any other codeword $\bldc'' \in \cC^*$ we have that
$\distance(\bldc,\bldc'') \ge 4$ and $\distance(\bldc',\bldc'') \ge 4$.
There are exactly $2^{2^r-r-1}$ such pairs of codewords
$\bldc,\bldc' \in \cC$ such that $\distance(\bldc,\bldc')=2$.
\end{lemma}

Similarly to NP1CCs, two codewords in an ENP1CC that are
at distance~$2$ apart will be called \Emph{partners}.

\begin{corollary}
The codewords of an ENP1CC $\cC^*$ can be partitioned into
$2^{2^r-r-1}$ pairs of partners.
\end{corollary}

Corollary~\ref{cor:sufficient} and Lemma~\ref{lem:part_d=2} imply
the following consequence.

\begin{corollary}
Puncturing an ENP1CC on any one of its coordinates yields an NP1CC.
\end{corollary}

A necessary and sufficient condition that a puncturing of
an ENP1CC will be of a certain type of an NP1CC
can be inferred as an immediate observation from
the definitions of \Type{\TA}, \Type{\TB}, and \Type{\TC}.

\begin{lemma}
\label{lem:punctABC}
Let $\cC^*$ is an ENP1CC.
\begin{itemize}
\item[(1)]
The punctured code of $\cC^*$ is an NP1CC of \Type{\TA},
if and only if in each pair of partners,
the partners disagree on the punctured coordinate.

\item[(2)]
The punctured code of $\cC^*$ is an NP1CC of \Type{\TB},
if and only if in each pair of partners, the partners agree on the punctured coordinate.

\item[(3)]
The punctured code of $\cC^*$ is an NP1CC of \Type{\TC}, if and only if in some of the pairs of partners,
the partners agree on the punctured coordinate while in some other pairs they disagree on that coordinate.
\end{itemize}
\end{lemma}

We will consider now which NP1CCs can be obtained from
one ENP1CC. We are interested to know if there are ENP1CCs
whose punctured codes are only of one type
or rather a combination of two or all three type.
This can be used to form equivalence classes among
the ENP1CCs and also among the NP1CCs.
In the rest of this section we consider these problems.

Lemmas~\ref{lem:part_d=2} and~\ref{lem:punctABC} immediately imply
the following consequence.
\begin{corollary}
\label{cor:no_Punct_AB}
$~$
\begin{itemize}
\item[(1)]
There are no ENP1CCs whose punctured codes are only of \Type{\TA}.

\item[(2)]
There are no ENP1CCs whose punctured codes are only of \Type{\TB}.
\end{itemize}
\end{corollary}

\begin{lemma}
\label{lem:PE_AB}
If $\cC$ is an NP1CC obtained from the union of
an extended zeroed perfect code $\cC_1$
and an odd translate $\cC_2$ of $\cC_1$,
then $\cC^*$ is an ENP1CC whose punctured codes are of
\Type{\TA} and \Type{\TB}.
\end{lemma}

\begin{proof}
Noting that $\cC_2 = \blde + \cC_1$ where $\weight(\blde)=1$,
the partners in each pair disagree on exactly one coordinate,
and that coordinate is the same for all pairs.
Therefore, in the extended code $\cC^*$,
the partners in each pair disagree
on this coordinate and on the new coordinate
and agree on the remaining $2^r-1$ coordinates.
Thus, by Lemma~\ref{lem:punctABC}(1), puncturing
on one of these two coordinates yields an NP1CC of \Type{\TA},
while by Lemma~\ref{lem:punctABC}(2), puncturing
on any of the other $2^r-1$ coordinates yields an NP1CC of \Type{\TB}.
\end{proof}

It is easy to verify by Lemma~\ref{lem:punctABC} that
all ENP1CCs whose punctured codes are of \Type{\TA} and \Type{\TB}
can be obtained by Lemma~\ref{lem:PE_AB}.

\begin{lemma}
\label{lem:PE_AC}
If $\cC$ is an NP1CC of \Type{\TA} in which
for each coordinate there exists
at least one pair of partners that disagree on that coordinate,
then the punctured code of $\cC^*$ are of \Type{\TA} and \Type {\TC}.
\end{lemma}

\begin{proof}
If $\cC$ is such an NP1CC, then for each coordinate
there is at least one pair of partners that disagree on that coordinate
and, since $\cC$ is of \Type{\TA}, it follows that in $\cC^*$,
in each pair, the partners disagree on the new coordinate.
Puncturing on the new coordinate yields
the original code of \Type{\TA} and, by Lemma~\ref{lem:punctABC}(3),
puncturing on any other coordinate yields an NP1CC of \Type{\TC}.
\end{proof}

We note that a balanced NP1CC is an NP1CC of \Type{\TA} which satisfies
the requirements of Lemma~\ref{lem:PE_AC}.
It is easy to verify by Lemma~\ref{lem:punctABC}
that all ENP1CCs whose punctured codes are of \Type{\TA} and \Type{\TC}
can be obtained by Lemma~\ref{lem:PE_AC}.

\begin{lemma}
\label{lem:cond_ABC}
In an ENP1CC whose punctured codes are of
\Type{\TA}, \Type{\TB}, and \Type{\TC}
there is exactly one coordinate on which
the partners disagree in all pairs,
and at least one coordinate on which all the partners agree.
\end{lemma}

\begin{proof}
By Lemma~\ref{lem:punctABC}(1), the punctured ENP1CC is
an NP1CC of \Type{\TA}, if and only if there exists
one coordinate on which the partners in each pair disagree.
By Lemma~\ref{lem:punctABC}(2), the punctured ENP1CC is
an NP1CC of \Type{\TB}, if and only if there exists
one coordinate on which the partners in each pair agree.
Finally, by Lemma~\ref{lem:punctABC}(3),
there exists at least one coordinate on which
partners in some pairs agree while in some other pairs disagree;
hence, there exists exactly one coordinate on which
the partners in each pair disagree.
\end{proof}

The conditions of Lemma~\ref{lem:cond_ABC} are necessary,
but they are also sufficient.
We construct such an ENP1CC based on
an idea presented in~\cite{EtVa94}.
By~\cite{EtVa94}, there exist two zeroed perfect codes of
length $2^r-1$ which differ only in $2^{2^{r-1}-1}$
codewords and only on one coordinate, say the first coordinate.
Let $\cC_1$ be the extended code of the first code and $\cC_2$ be
an odd translate of the extended code for the second
(where the extended code and its translate differ only on
the last coordinate).

\begin{lemma}
\label{lem:PE_ABC}
The ENP1CC $\cC^*$ obtained by extending
the code $\cC \triangleq \cC_1 \cup \cC_2$ is an ENP1CC
whose punctured codes are of \Type{\TA}, \Type{\TB}, and \Type{\TC}.
\end{lemma}

\begin{proof}
Clearly, $\cC^*$ has
one coordinate on which the partners in each pair disagree;
two coordinates on which there is agreement in some of the pairs;
and $2^r-2$ coordinates on which the partners in each pair agree.
The result follows from Lemma~\ref{lem:punctABC}.
\end{proof}

Corollaries~\ref{cor:no_Punct_AB} and Lemmas~\ref{lem:PE_AB},~\ref{lem:PE_AC}, and~\ref{lem:PE_ABC} raise the question
whether there exists an ENP1CC with no punctured code of \Type{\TA}.

We end this section by a characterization of
the weight enumerator of a zeroed ENP1CC.
Interestingly, this weight distribution turns out to be unique
and independent of the type of the NP1CC that was extended
(this also implies that ENP1CCs are distance invariant).

\begin{theorem}
\label{thm:weightenumerator-ENP1CC}
Let $\cC^*$ be a zeroed $(n{+}1{=}2^r{+}1,M{=}2^{n-r})$~ENP1CC.
Its weight enumerator is given by
\[
\AA^*(y)
= \frac{1}{2n} \left( (1+y)^{n+1} + (1-y)^{n+1} \right)
+ \left( 1 - \frac{1}{n} \right) (1-y^2)^{n/2} ~.
\]
\end{theorem}

\begin{proof}
Let $\cC$ be the zeroed $(n,M)$ NP1CC that was extended
and let $\AA(y) = \sum_{i \in \Int{0:n}}^n A_i y^i$ be
its weight enumerator. It is easy to see
that the weight distribution of $\cC^*$ is given by
\[
A^*_0 = 1 ~, \quad A^*_{n+1} = 0 ~,
\]
and, for $i \in \Int{n}$:
\[
A^*_i =
\left\{
\begin{array}{ccl}
A_i + A_{i-1} && \textrm{if $i$ is even} \\
0             && \textrm{otherwise.}
\end{array}
\right.
\]
Hence,
\begin{eqnarray}
\AA^*(y)
= \sum_{i \in \Int{0:n+1}} A^*_i y^i
& = &
\frac{1}{2}
\bigl(
\AA(y) + \AA(-y)
+ y(\AA(y) - \AA(-y))
\bigr)
\nonumber \\
\label{eq:weightenumerator-ENP1CC}
& = &
\frac{1}{2}
\bigl( (1+y) \AA(y) + (1-y) \AA(-y) \bigr) ~.
\end{eqnarray}
Substituting either~(\ref{eq:TypeA}) or~(\ref{eq:TypeB})
into~(\ref{eq:weightenumerator-ENP1CC}) yields the result.
\end{proof}

\section{Conclusion and Future Work}
\label{sec:conclude}

The structure of NP1CCs was considered.
It was proved that there are three types of such codes which depend on the
distance between each codeword to its nearest codeword.
The structure of these codes, their weight and distance distributions are examined in the paper.
Constructions of a large number of codes of each type were given. The extended code of an NP1CC was analyzed
and in particular it was discussed which types of NP1CCs are obtained by puncturing each of its coordinates.
Our exposition leads to a many interesting open problems.

\begin{enumerate}
\item Is it true that there exist two perfect codes of length $2^r-1$ and intersection $k$ if and only if there exists
an NP1CC of \Type{\TC} with exactly $k$ \Type{\Tone} pairs? What is the minimum (maximum) possible number of \Type{\Tone} pairs in
an NP1CC of \Type{\TC}?

\item Let $\cX$ and $\cY$ be two distinct nonempty sets of pairwise disjoint capsules such that
$$
\bigcup_{V \in \cX} V = \bigcup_{V \in \cY} V.
$$
What is the minimum size of $\cX$ and $\cY$?

\item Does there exist an NP1CC of \Type{\TB} in which for each pair of coordinates there exist at least one pair of partners whose
partners disagree on this pair of coordinate?

\item We proved that there exists a balanced NP1CC of \Type{\TA}.
Does there exist a similar code of \Type{\TB}?
One possible definition for balanced NP1CCs of type B is that a pair of \Type{\Ttwo} pairs disagree
only on coordinates $i$ and $i+1$, $1 \leq i \leq 2^r-1$ or on coordinates $1$ and $2^r$ and the number of such partner pairs
for these coordinates is the same. Are there balanced NP1CCs for this definition?

\item Does there exist an ENP1CC with no punctured code of \Type{\TA}?
In other words, does there exist an ENP1CCs whose punctured codes are only of \Type{\TC}? or
does there exist an ENP1CC whose punctured codes are only of \Type{\TB} and \Type{\TC}?
\end{enumerate}

\section*{Acknowledgement}

The authors would like to thank an anonymous reviewer for his comprehensive review and constructive suggestions

%

\bibliographystyle{IEEEtran}

\end{document}